\documentclass[a4paper,UKenglish,cleveref,autoref,thm-restate]{lipics-v2021}

\usepackage{subcaption}
\usepackage{xspace}

\graphicspath{{}}

\crefname{equation}{eq.}{eqs.}
\crefname{enumi}{}{}
\crefname{icase}{case}{cases}
\crefname{ipart}{part}{parts}
\crefname{iprop}{property}{properties}
\crefname{iinv}{invariant}{invariants}
\crefformat{section}{\S\,#2#1#3}
\crefformat{subsection}{\S\,#2#1#3}
\crefrangeformat{section}{\S\S\,#3#1#4--#5#2#6}
\crefmultiformat{section}{\S\S\,#2#1#3}{,~#2#1#3}{,~#2#1#3}{,~#2#1#3}

\newcommand{\fO}{\mathcal{O}}

\newcommand{\eps}{\varepsilon}
\newcommand{\NP}{\ensuremath{\mathsf{NP}}}
\newcommand{\Wone}{\ensuremath{\mathsf{W}[1]}}
\newcommand{\LPPM}{\textsc{Left PPM}\xspace}
\newcommand{\PPM}{\textsc{PPM}\xspace}
\newcommand{\PSI}{\textsc{PSI}\xspace}

\newcommand{\GapSPPM}{$(n^{\eps k}, n^{(1-\eps) \cdot k})$\textsc{-GAP \#PPM}\xspace}
\newcommand{\GapSPPMh}{$(n^{(\eps/2) \cdot k}, n^{(1-\eps/2) \cdot k})$\textsc{-GAP \#PPM}\xspace}

\title{Inapproximability of Counting Permutation Patterns}

\author{Michal Opler}
      {Faculty of Information Technology, Czech Technical University in Prague, Czech Republic}
      {michal.opler@fit.cvut.cz}
      {https://orcid.org/0000-0002-4389-5807}
      {}

\authorrunning{M. Opler}

\Copyright{Michal Opler}

\category{Track A: Algorithms, Complexity and Games}



\relatedversiondetails{ArXiv}{https://arxiv.org/abs/2601.05166}

\funding{This work was co-funded by the European Union under the project
Robotics and advanced industrial production (reg.\ no.\ CZ.02.01.01/00/22\_008/0004590).}

\EventEditors{Sayan Bhattacharya, Danupon Nanongkai, Michael Benedikt, and Gabriele Puppis}
\EventNoEds{4}
\EventLongTitle{53rd International Colloquium on Automata, Languages, and Programming (ICALP 2026)}
\EventShortTitle{ICALP 2026}
\EventAcronym{ICALP}
\EventYear{2026}
\EventDate{July 7--10, 2026}
\EventLocation{Royal Holloway, University of London, Egham, United Kingdom}
\EventLogo{}
\SeriesVolume{374}
\ArticleNo{8}

\hideLIPIcs

\begin{document}
\nolinenumbers

\maketitle

\begin{abstract}
Detecting and counting copies of permutation patterns are fundamental algorithmic problems, with applications in the analysis of rankings, nonparametric statistics, and property testing tasks such as independence and quasirandomness testing.
From an algorithmic perspective, there is a sharp difference in complexity between detecting and counting the copies of a given length-$k$ pattern in a length-$n$ permutation.
The former admits a $2^{\fO(k^2)} \cdot n$ time algorithm (Guillemot and Marx, 2014) while the latter cannot be solved in time $f(k)\cdot n^{o(k/\log k)}$ unless the Exponential Time Hypothesis (ETH) fails (Berendsohn, Kozma, and Marx, 2021).
In fact already for patterns of length~4, exact counting is unlikely to admit near-linear time algorithms under standard fine-grained complexity assumptions (Dudek and Gawrychowski, 2020).

Recently, Ben-Eliezer, Mitrovi\'c and Srivastava~(2026) showed that for patterns of length up to~5, a $(1+\varepsilon)$-approximation of the pattern count can be computed in near-linear time, yielding a separation between exact and approximate counting for small patterns, and conjectured that approximate counting is asymptotically easier than exact counting in general.
We strongly refute their conjecture by showing that, under ETH, no algorithm running in time $f(k)\cdot n^{o(k/\log k)}$ can approximate the number of copies of a length-$k$ pattern within a multiplicative factor $n^{(1/2-\varepsilon)k}$.
The lower bound on runtime matches the conditional lower bound for exact pattern counting, and the obtained bound on the multiplicative error factor is essentially tight, as an $n^{k/2}$-approximation can be computed in $2^{\mathcal{O}(k^2)}\cdot n$ time using an algorithm for pattern detection.
\end{abstract}

\section{Introduction}\label{sec:intro}

Given two permutations $\pi = \pi_1,\dots,\pi_k$ and $\tau = \tau_1,\dots,\tau_n$, we say that $\tau$ \emph{contains} $\pi$ as a pattern if there exist indices $1 \le i_1 < \cdots < i_k \le n$ such that the subsequence $\tau_{i_1},\dots,\tau_{i_k}$ has the same relative order as $\pi$.
Otherwise, we say that $\tau$ \emph{avoids}~$\pi$.
Permutation patterns have been the subject of extensive study in enumerative combinatorics, with deep connections to sorting algorithms, forbidden substructures, and the structural theory of permutations.

The algorithmic problem associated with this notion is \textsc{Permutation Pattern Matching} (\PPM), which asks, given a pattern permutation $\pi$ and a text permutation $\tau$, whether $\tau$ contains~$\pi$.
A related counting variant asks for the number of occurrences of $\pi$ in~$\tau$.
The computational complexity of \PPM{} and its counting variant has been studied extensively under various restrictions, which we briefly review below.

\subparagraph{Pattern matching.}
In the most general form, it has been shown by Bose, Buss and Lubiw~\cite{Bose_PPM} that \PPM is \NP-complete in general.
There are two natural restrictions that can be imposed to restore tractability:
either restricting the pattern or the text to a fixed class of permutations, or focusing on instances where $k$, the length of the pattern, is small.
In the former case, the natural restriction is to fix a permutation $\sigma$ and only consider patterns (or texts) that themselves do not contain~$\sigma$.
Jelínek and Kynčl~\cite{JelinekK17} provided a full complexity dichotomy for \PPM with $\sigma$-avoiding patterns and some initial results for the case when the text is also required to avoid $\sigma$.
Later, Jelínek, Opler and Pekárek~\cite{JelinekOP21} resolved the complexity of \PPM with $\sigma$-avoiding texts up to essentially four remaining open cases.

In the case of small~$k$, there is a trivial algorithm that solves both pattern matching and pattern counting in $\fO(k \cdot n^k)$ time simply by enumerating over all possible $k$-tuples of elements in the text.
In a breakthrough result, Guillemot and Marx~\cite{GM_PPM} showed that permutation pattern matching can be solved in $2^{\fO(k^2 \log k)} \cdot n$ time, with a later improvement to $2^{\fO(k^2)} \cdot n$ due to Fox~\cite{jfox}, establishing that \PPM{} is fixed-parameter tractable with respect to $k$.
The algorithm uses a win-win argument based on the celebrated proof of the Füredi-Hajnal conjecture~\cite{FurediHajnal} by Marcus and Tardos~\cite{MarcusTardos}.

\subparagraph{Pattern counting.}
Unlike detection, the aforementioned algorithm of Guillemot and Marx crucially cannot be adapted to pattern counting.
Instead, there has been a line of work improving upon the trivial $\fO(k \cdot n^k)$-time algorithm.
First, Albert, Aldred, Atkinson and Holton~\cite{AlbertAAH01} designed an algorithm for pattern counting in time $n^{2k/3 + o(k)}$.
Their ideas were later further developed by Ahal and Rabinovich~\cite{AhalR08} to count in time $n^{0.47k + o(k)}$.
The current best algorithm in this regime is due to Berendsohn, Kozma and Marx~\cite{BerendsohnKM21} and runs in time $n^{k/4 + o(k)}$.
Significant improvements of the runtime are unlikely as the same authors showed that no algorithm for pattern counting runs in time $f(k)\cdot n^{o(k/\log k)}$ unless the exponential-time hypothesis (ETH) fails.
Remarkably, only a slightly weaker conditional lower bound holds even when we impose additional structural restrictions on the patterns.
Jelínek, Opler and Pekárek~\cite{JelinekOP21a} showed that for arbitrary $\sigma$ of length at least~6, no algorithm for counting $\sigma$-avoiding patterns runs in time $f(k) \cdot n^{o(k/\log^2 k)}$ unless ETH fails.


\subparagraph{Counting short patterns.}
Independently of the parameterized regime, a line of work focused on the regime where $k$ is a small constant.
For $k=2$, this is the classical problem of counting the number of inversions in a sequence.
A classical $\fO(n \log n)$-time solution can be obtained by a modification of merge sort while the fastest known algorithm (in the Word RAM model) with runtime $\fO(n \sqrt{\log n})$ is due to Chan and P{\u{a}}tra{\c{s}}cu~\cite{ChanP10}.
A wide-range of problems can be reduced to counting inversions, including string processing and computational geometry problems, thus inheriting the same $\fO(n \sqrt{\log n})$ runtime \cite{KempaK19,KempaK25}.

Even-Zohar and Leng~\cite{Even-ZoharL21} introduced a novel dynamic programming approach to counting patterns which allowed them to count all 3-patterns and some (8 out of 24) 4-patterns in $\tilde{\fO}(n)$ time\footnote{The $\tilde{\fO}(\cdot)$ notation hides polylogarithmic factors.}.
Dudek and Gawrychowski~\cite{DudekG20} proved that counting these ``hard'' 4-patterns is actually equivalent (via bidirectional reductions in near-linear time) to counting 4-cycles in sparse graphs and thus, it requires $n^{1+\Omega(1)}$ time conditional on, e.g., the Strong 3-SUM conjecture~\cite{JinX23}.
Recently, Beniamini and Lavee~\cite{BeniaminiL25} expanded upon the ideas of Even-Zohar and Leng to obtain algorithms for counting 5-patterns in $\tilde{\fO}(n^{7/4})$ time and 6-patterns and 7-patterns in $\tilde{\fO}(n^2)$ time.

\subparagraph{Approximate counting.}
Since exact counting is unlikely to be solvable in near-linear time already for $k=4$, it is natural to consider the relaxation to approximate counting.
The most desired outcome is an efficient $(1+\eps)$-approximation algorithm for each fixed positive $\eps$, i.e., an algorithm that reports a value between $C/(1+\eps)$ and $(1+\eps) \cdot C$ where $C$ is the true number of $\pi$-copies in $\tau$.
In the case of $k=2$, there has been a long line of work focusing on approximating the number of inversions that culminated with an $\fO(n)$-time $(1+\eps)$-approximation algorithm by Chan and P{\u{a}}tra{\c{s}}cu~\cite{ChanP10}.
For $k=3$, we have exact counting algorithms in near-linear time so there is little room for improvement with approximation.
This was the whole picture until very recently when Ben-Eliezer, Mitrovi\'c and Srivastava~\cite{BenEliezerMS26} obtained $\tilde{\fO}(n)$-time $(1+\eps)$-approximation algorithms for counting all 4-patterns and 5-patterns.
This yields a separation between the complexities of exact and approximate counting already for $k=4$.
Notably, they also conjectured that the time complexity of approximate counting is asymptotically smaller than that of exact counting~\cite[Conjecture 1.5]{BenEliezerMS26}.

\subparagraph{Hard variants of pattern matching.}
While detecting unrestricted $k$-patterns admits a linear time algorithm for every fixed~$k$, i.e., an FPT algorithm with respect to~$k$, a similarly efficient algorithm cannot exist for counting $k$-patterns unless ETH fails.
We might, therefore, ask what happens when we consider, instead of counting patterns, detecting patterns with some additional constraints.
It turns out that \PPM becomes \Wone-hard even under mild additional constraints and thus is unlikely to admit FPT algorithms with respect to $k$.

Several such constraints were explored by Bruner and Lackner~\cite{BrunerL2013}.
As an intermediate problem, they proved \Wone-hardness of \textsc{Segregated PPM} where the elements of both pattern and text are partitioned by value into small and large, and the task is to find an embedding which maps small elements to small elements and large elements to large elements only.
Consequently, Bruner and Lackner~\cite[Theorem 5.5]{BrunerL2013} deduced \Wone-hardness of \textsc{Vincular PPM} where some elements of the pattern are required to map to consecutive elements in the text.
Upon closer inspection, their reduction produces instances where only the first element of the pattern is required to map to the first element of the text.
We refer to such copies as left-aligned and they play a crucial role in our result.

In a different direction, Guillemot and Marx~\cite{GM_PPM} showed that a 3-dimensional variant of \PPM is \Wone-hard and thus, their FPT algorithm cannot be extended to permutations in higher dimensions.
Other \Wone-hard variants of \PPM include \textsc{Partitioned PPM}, where each element of the pattern has prescribed possible locations in the text~\cite{BerendsohnKM21,GM_PPM}, and its slightly more relaxed variant \textsc{Surjective Colored PPM}~\cite{Berendsohn2019}.

\subsection*{Our contribution}

We show that there is no FPT algorithm for approximate pattern counting under the exponential-time hypothesis, even if we allow the multiplicative error to be as large as $n^{(1/2-\eps) \cdot k}$ for arbitrarily small positive $\eps$.
In fact, we rule out the existence of such an approximation algorithm with runtime $f(k) \cdot n^{o(k/\log k)}$, matching the conditional lower bound on exact counting due to Berendsohn, Kozma and Marx~\cite[Theorem 4]{BerendsohnKM21}.
This strongly refutes the conjecture by Ben-Eliezer, Mitrovi\'c and Srivastava~\cite[Conjecture 1.5]{BenEliezerMS26}.

\begin{restatable}{theorem}{mainresult}
\label{thm:inapprox}
	For arbitrary $0 < \eps < 1/2$, an algorithm computing the number of copies of a given $k$-pattern with $n^{(1/2-\eps) \cdot k}$-multiplicative error in $f(k) \cdot n^{o(k/\log k)}$ time would refute ETH.
\end{restatable}

Furthermore, the bound on the multiplicative error cannot be significantly improved since we can easily compute an $n^{1/2 \cdot k}$-approximation in $2^{\fO(k^2)} \cdot n$ time.
It suffices to invoke the FPT algorithm for \PPM by Guillemot and Marx~\cite{GM_PPM} and output~$0$ if the text does not contain the pattern at all, and $n^{1/2\cdot k}$ otherwise.

Our results are fully self-contained.
In \cref{sec:lppm}, we show that under ETH there is no $f(k) \cdot n^{o(k/\log k)}$-time algorithm for detecting left-aligned copies (the \LPPM problem)
\footnote{This reduction previously appeared in the author’s PhD thesis~\cite[Proposition 4.7]{Opler2022}.}.
In \cref{sec:inapprox}, we construct a gap-producing reduction from \LPPM that, on input $(\pi, \tau)$, produces an output $(\pi', \tau')$ such that $\tau'$ contains a very large number of $\pi'$-copies if there is a left-aligned $\pi$-copy in~$\tau$ and otherwise, $\tau'$ contains very few $\pi'$-copies.
The inapproximability of pattern counting (\cref{thm:inapprox}) then follows via standard arguments.

\section{Preliminaries}\label{sec:prelim}
\subparagraph{Permutations and point sets.}A \emph{permutation of length $n$} (or just \emph{$n$-permutation}) is a sequence $\pi = \pi_1, \dots, \pi_n$ in which each element of the set $[n] = \lbrace 1, 2, \dots, n\rbrace$ appears exactly once.
As customary, we omit commas and write, e.g., $15342$ for the permutation $1,5,3,4,2$ when there is no ambiguity.
It is often convenient to view permutations as point sets, namely as the point set $S_\pi = \{(i,\pi_i);\;i\in[n]\}$ in the plane.
See Figure~\ref{fig:perm}.
We refer to $S_\pi$ as the \emph{diagram of~$\pi$} and we freely move between sequence-based representation of permutations and their diagrams.
Observe that no two points in a permutation diagram share the same $x$- or $y$-coordinate.
We say that such a set is in \emph{general position}.

For a point $p$ in the plane, we let $p.x$ denote its horizontal coordinate, and $p.y$ its vertical coordinate.
Two finite sets $S, R \subseteq \mathbb{R}^2$ in general position are \emph{isomorphic} if there exists a bijection $f\colon S \to R$ that preserves the relative order of $x$- and $y$-coordinates, that is,
$f(p).x < f(q).x \Leftrightarrow p.x < q.x$ and
$f(p).y < f(q).y \Leftrightarrow p.y < q.y$
for all $p,q\in S$.
The \emph{reduction} of a finite set $S \subseteq \mathbb{R}^2$ in general position is the unique permutation $\pi$ such that $S$ is isomorphic to $S_\pi$.

We say that an $n$-permutation $\tau$ \emph{contains} a $k$-permutation $\pi$ (also referred to as a $k$-pattern) if the diagram of $\tau$ contains a subset that is isomorphic to the diagram of~$\pi$.
The witnessing injective function $f: S_\pi \to S_\tau$ is called an \emph{embedding of $\pi$ into $\tau$.}
We sometimes refer to the image $f(S_\pi)$ as the \emph{copy of $\pi$} (or just \emph{$\pi$-copy}).
See Figure~\ref{fig:perm}.
Additionally, we say that a $\pi$-copy (and the corresponding embedding) is \emph{left-aligned} if the leftmost point of $S_\pi$ (namely $(1,\pi_1)$) is mapped to the leftmost point of $S_\tau$ (namely $(1,\tau_1)$).
See Figure~\ref{fig:pattleft}.

\begin{figure}
	\captionsetup[subfigure]{justification=centering}
	\begin{subfigure}{0.2\textwidth}
			\centering
			\includegraphics[width=\textwidth]{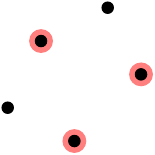}
			\caption{}\label{fig:perm}
	\end{subfigure}
	\hfill
	\begin{subfigure}{0.2\textwidth}
		\centering
		\includegraphics[width=\textwidth]{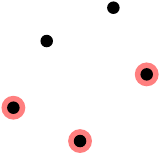}
		\caption{}\label{fig:pattleft}
	\end{subfigure}
	\hfill
	\begin{subfigure}{0.215\textwidth}
			\centering
			\includegraphics[width=\textwidth]{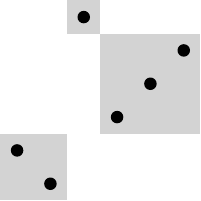}
			\caption{}\label{fig:infl}
	\end{subfigure}
	\hfill
	\begin{subfigure}{0.215\textwidth}
		\centering
		\includegraphics[width=\textwidth]{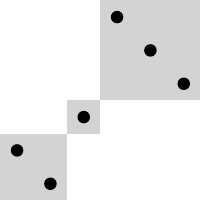}
		\caption{}\label{fig:layered}
	\end{subfigure}
	\caption{(a) Permutation $24153$ with a highlighted copy of the pattern $312$, (b) a left-aligned copy of the pattern $213$, (c) the inflation of $132$ by $21$, $1$ and $123$, and (d) a layered permutation with highlighted layers.}
\end{figure}

\subparagraph{Inflations, layered and co-layered permutations.}
Given an $n$-permutation $\sigma$ and $n$ non-empty permutations $\tau_1, \ldots, \tau_n$ the \emph{inflation} of $\sigma$ by $\tau_1, \dots, \tau_n$ is the permutation obtained by replacing each point $(i, \sigma_i)$ in the diagram of $\sigma$ with a suitably scaled copy of the diagram of $\tau_i$ and then taking the reduction of the obtained point set.
See Figure~\ref{fig:infl}.
We say that a permutation~$\pi$ is \emph{layered} if it is the inflation of an increasing $n$-permutation $\iota$ by $n$ decreasing permutations $\delta_1,\dots, \delta_n$.
We refer to $\delta_1, \dots, \delta_n$ as the \emph{layers} of~$\pi$.
See Figure~\ref{fig:layered}.
Symmetrically, a permutation is \emph{co-layered} if it is an inflation of a decreasing $n$-permutation $\delta$ with $n$ increasing permutations $\iota_1, \dots, \iota_n$.

\subparagraph{Parameterized complexity.}
Parameterized complexity provides a refined view of computational hardness by measuring complexity not only in terms of the input size but also with respect to a designated parameter\footnote{In our case, the parameter will be exclusively the length of the pattern.}.
An algorithm is called \emph{fixed-parameter tractable (FPT)} with respect to a parameter~$k$ if it runs in time $f(k)\cdot n^{\fO(1)}$, where $n$ denotes the input size and $f$ is a computable function.
Analogously to classical complexity theory, there are parameterized notions of intractability: problems that are \Wone-hard under parameterized reductions are widely believed not to admit FPT algorithms.
Lower bounds in parameterized complexity are often based on the \emph{Exponential Time Hypothesis (ETH)}, which roughly asserts that \textsc{3-SAT} cannot be solved in subexponential time with respect to the number of variables.
For a comprehensive introduction to the area, we refer the reader to the monograph by Cygan et al.~\cite{CyganFKLMPPS15}.

\section{Hardness of detecting left-aligned copies}
\label{sec:lppm}

We first show that it is hard to detect left-aligned patterns.
Formally, we define the problem \LPPM where we receive a $k$-pattern~$\pi$ and an $n$-permutation~$\tau$ (referred to as the text) as input and the task is to decide whether $\tau$ contains a left-aligned $\pi$-copy.
The work of Bruner and Lackner~\cite[Theorem 5.5]{BrunerL2013} implies the \Wone-hardness of \LPPM but, crucially, not the lower bound under ETH.

\begin{theorem}
	\label{thm:appm-hard}
	\LPPM is \Wone-hard with respect to $k$, and unless ETH fails, it cannot be solved in time $f(k) \cdot n^{o(k/ \log k)}$ for any function $f$, where $k$ is the length of the pattern and $n$ is the length of the text.
\end{theorem}
\begin{proof}[Proof of \cref{thm:appm-hard}]
	We reduce from the well-known problem \textsc{Partitioned Subgraph Isomorphism} (\PSI) sharing the basic structure with the reduction from \PSI to \textsc{Partitioned PPM} by Berendsohn, Kozma and Marx~\cite{BerendsohnKM21}.
	The input to \PSI consists of two graphs $G = (V_G, E_G)$ and $H = (V_H, E_H)$ together with a coloring $\chi\colon V_H \to V_G$ of vertices of $H$, using the vertices of $G$ as colors; and the task is to decide if there is a mapping $\phi\colon V_G \to V_H$ such that whenever $\{u,v\} \in E_G$ then also $\{\phi(u), \phi(v)\} \in E_H$ and moreover, $\chi(\phi(v)) = v$ for every $v \in V_G$.
	The problem \PSI{} is \Wone-complete with respect to $k = |E_G|$ and Marx~\cite[Cor. 6.3]{Marx10} 
	showed that it cannot be solved in time $f(k) \cdot n^{o(k/ \log k)}$ unless ETH fails.
	Moreover, this holds even if we require $G$ to have the same number of vertices and edges (see e.g. \cite{Bringmann0MN16}).
	
	Let $(G,H,\chi)$ be an instance of \PSI where $|V_G| = |E_G|$ and set $n = |V_H|, k = |V_G|$ (refer to Figure~\ref{fig:appm-hard-instance}).
	We assume that the vertex set $V_G$ is in fact equal to $[k]$ and we define for each $i \in[k]$ the set $V_i \subseteq V_H$ as the set of vertices of $H$ colored by $i$, i.e., $V_i = \chi^{-1}(i)$.
	Notice that $V_1, \ldots, V_k$ form a partition of the set $V_H$.
	
	We shall construct two point sets $P$ and $T$ not necessarily in general position such that $P$ will represent the adjacency matrix of $G$ while $T$ will represent the adjacency matrix of $H$.
	Afterwards, the permutations $\pi$ and $\tau$ are obtained as reductions of a small clockwise rotation of the point sets $P$ and $T$ respectively.
	Observe that a sufficiently small rotation preserves all relative orderings between pairs of elements with non-equal coordinates.
  Alternatively, we can achieve the same relative position of points by mapping each point $(x,y)$ to a point $(x + \delta \cdot y, y - \delta\cdot x)$ for sufficiently small but positive~$\delta$.
	
	\subparagraph{Constructing the pattern~$\pi$.}
	We start with the description of the set $P$.
	It contains two points $a^P_1$ and $a^P_2$ defined as
	\[a^P_1 = (1, 2k + 2), \qquad a^P_2 = (2k+2, 1).\]
	We refer to $a^P_1$ and $a^P_2$ as the \emph{anchors}.

	For each $i \in [k]$, we associate two pairs of points to the vertex $i$ defined as
	\begin{align*}
		A_i &= \{ (2i, 3i + 2k), (2i + 1, 3i + 2k + 2)\}, \\
		B_i &= \{ (3i + 2k, 2i), (3i + 2k + 2, 2i + 1)\}.
	\end{align*}
	Every pair $A_i$ lies horizontally between the anchors and every pair $B_i$ lies vertically between the anchors.
	The pairs $A_1, \dots, A_k$ form together an increasing permutation of length $2k$ and the same holds for the pairs $B_1, \dots, B_k$.
	The pairs naturally impose a grid-like structure.
	We define the \emph{$A_i$-row} as the horizontal strip enclosed by the pair $A_i$, the \emph{$B_j$-column} as the vertical strip enclosed by $B_j$ and the \emph{$(A_i, B_j)$-cell} as their intersection.

	For each edge $\{i,j\} \in E_G$, we simply add points to the $(A_i, B_j)$-cell and the $(A_j, B_i)$-cell, i.e., we add to $P$ the points
	\[ (3i + 2k+1, 3j + 2k+1), \quad (3j + 2k+1, 3i + 2k+1).\]
	Additionally, we also add a point to each cell on the diagonal, i.e., we add the point $(3i + 2k+1, 3i + 2k+1)$ for every $i$.
	
	That wraps up the definition of $P$.
	We rotate $P$ clockwise slightly to guarantee that it is in general position and take the permutation $\pi$ as its reduction.
	See Figure~\ref{fig:appm-hard-pattern}.
	The length of $\pi$ is $\fO(|V_G| + |E_G|) \in \fO(k)$ since we assumed that $|V_G| = |E_G|$.
	
	\begin{figure}
		\begin{subfigure}[c]{0.44\textwidth}
		\centering
		\begin{subfigure}{\textwidth}
			\centering
			\includegraphics{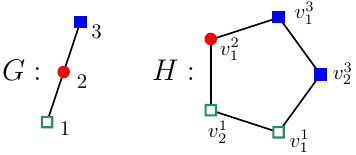}
			\caption{Instance $(G, H, \chi)$ of \PSI.}\label{fig:appm-hard-instance}
		\end{subfigure}
		\begin{subfigure}{\textwidth}
			\centering
			\vspace{0.22in}
      \includegraphics[scale=0.975]{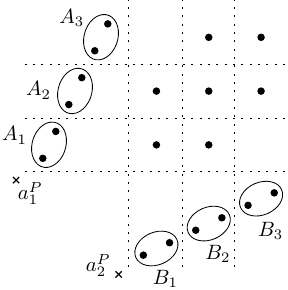}
			\caption{The produced point set $P$.}\label{fig:appm-hard-pattern}
		\end{subfigure}
		\end{subfigure}
		\begin{subfigure}[c]{0.55\textwidth}
			\centering
      \includegraphics[scale=0.975]{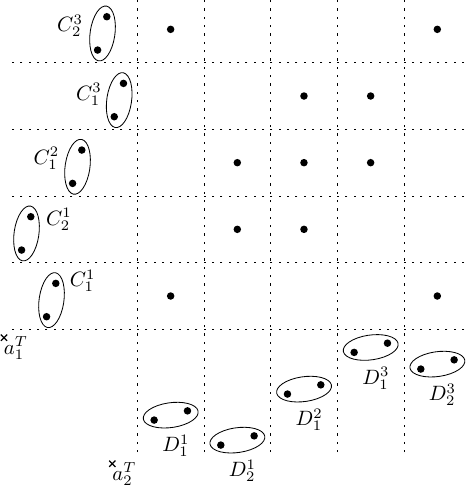}
			\caption{The produced point set $T$.}\label{fig:appm-hard-text}
		\end{subfigure}
		\caption{Illustration of the reduction in \cref{thm:appm-hard}. The permutations $\pi$ and $\tau$ are obtained as reductions of a small clockwise rotation of the point sets $P$ and $T$ respectively.}
		\label{fig:appm-hard}
	\end{figure}
	
	\subparagraph{Constructing the text~$\tau$.}
	Now, we shift our attention to the point set $T$.
	It again contains two anchors $a^T_1$ and $a^T_2$ defined as
	\[a^T_1 = (1, 2n + 2), \qquad a^T_2 = (2n+2, 1).\]
	
	For each $i \in [k]$, we set $n_i = |V_i|$ and choose an arbitrary order of vertices in $V_i$ denoting them $v^i_j$ for $j \in [n_i]$.
	To every vertex $v^i_j$, we associate two values -- the \emph{rank of $v^i_j$} denoted by $\alpha^i_j$ and the \emph{reverse rank of $v^i_j$} denoted by $\beta^i_j$ where
	\begin{equation}
		\label{eq:rank}
		\alpha^i_j = \sum_{i' < i} n_{i'} + j-1 \quad \text{ and } \quad \beta^i_j = \sum_{i'< i} n_{i'} + n_i -j.
	\end{equation}
	Observe that the rank corresponds to the lexicographic order of $v^i_j$ by $(i,j)$ and the reverse rank corresponds to the lexicographic order by $(i, n_i - j)$.
	
	For every $i \in [k]$ and $j \in [n_i]$, we add to $T$ two pairs of points associated to the vertex $v^i_j$
	\begin{align*}
		C^i_j &= \{ (2 \beta^i_j, 3\alpha^i_j + 2k), (2 \beta^i_j + 1, 3\alpha^i_j + 2k + 2)\}, \\
		D^i_j &= \{ (3\alpha^i_j + 2k, 2\beta^i_j), (3\alpha^i_j + 2k + 2, 2\beta^i_j + 1)\}.
	\end{align*}
	Every pair $C^i_j$ again lies horizontally between the anchors while every pair $D^i_j$ lies vertically between the anchors.
	For a fixed $i$, the pairs $C^i_1, \ldots, C^i_{n_i}$ form a co-layered permutation with each layer consisting of a single pair, and the same holds for the pairs $D^i_1, \ldots, D^i_{n_i}$.
	Moreover for different $i < j$, the pairs $C^i_1, \ldots, C^i_{n_i}$ lie all to the left and below the pairs $C^j_1, \ldots, C^j_{n_j}$.
	The same holds for the pairs $D^i_j$.
	
	We define $D^i_j$-columns, $C^{i'}_{j'}$-rows and $(D^i_j, C^{i'}_{j'})$-cells analogously to before. 
	Finally, we add a point to the $(D^i_j, C^{i'}_{j'})$-cell for $i, i' \in [k]$ and $j \in [n_i], j' \in [n_{i'}]$ whenever either $i = i'$ and $j = j'$ (i.e. on the main diagonal), or $i \neq i'$ and $\{v^i_j, v^{i'}_{j'}\} \in E_H$.
	Formally, such a point is defined as $(3\alpha^i_j + 2k+1, 3\alpha^{i'}_{j'} + 2k+1)$.
	In other words, every non-empty cell either lies on the diagonal or corresponds to an edge between two vertices that do not share the same color.
	
	Finally, we rotate $T$ clockwise to guarantee general position and take $\tau$ as its reduction.
	See Figure~\ref{fig:appm-hard-text}.
	The length of $\tau$ is $\fO(|V_H| + |E_H|)$ which is clearly bounded by $\fO(n^2)$.
	
	\subparagraph{Correctness (``only if'').}
	Suppose that $(G,H,\chi)$ is a yes-instance of \PSI.
	There is a witnessing mapping $\phi\colon [k] \to \mathbb{N}$ such that $\phi(i) \in [n_i]$ for 
	every $i \in [k]$ and $\{v^i_{\phi(i)}, v^j_{\phi(j)}\} \in E_H$ for every different $i, j \in [k]$ such that $\{i,j\}\in E_G$.
	
	We define a left-aligned embedding $\psi$ of $\pi$ into  $\tau$ by mapping the elements of~$\pi$ as follows.
	First, we take care of the left-aligned property by mapping the anchors in~$\pi$ to the anchors in~$\tau$.
	We map the pair $A_i$ to $C^i_{\phi(i)}$ and $B_i$ to $D^i_{\phi(i)}$ for each $i \in [k]$.
	It is sufficient to argue that every non-empty $(A_i, B_j)$-cell in $\pi$ maps to a non-empty cell in $\tau$.
	This follows immediately for the cells in $\pi$ on the diagonal.
	Otherwise if $i \neq j$, we have $\{i, j\} \in E_G$ so there must be an edge $\{v^i_{\phi(i)}, v^j_{\phi(j)}\}$ in~$H$ and thus, the $(D^j_{\phi(j)},C^i_{\phi(i)})$-cell is non-empty.
	
	\subparagraph{Correctness (``if'').}
	Suppose there exists a left-aligned embedding of $\pi$ into~$\tau$.
	The key idea is that the anchors and monotone sequences force the embedding to respect the grid structure induced by the pairs $A_i,B_i$ and $C^i_j,D^i_j$.
	We claim that any embedding of $\pi$ into $\tau$ that maps the anchor $a^P_1$ to the anchor $a^T_1$ must also map $a^P_2$ to the anchor $a^T_2$.
	This holds since the points of $\pi$ below $a^P_1$ form an increasing sequence of length exactly $2k + 1$ starting with the point $a^P_2$ while any longest increasing sequence in $\tau$ below $a^T_1$ is also of length exactly $2k + 1$ and starts with the point $a^T_2$.
	
	The pairs $A_1, \dots, A_k$ form an increasing sequence of length $2k$ sandwiched horizontally between the anchors of $\pi$.
	Therefore, they must all be mapped to the union of all the pairs $C^i_j$ since these are the only points in the horizontal strip between the anchors in $\tau$.
	However, the only increasing subsequences of length $2k$ in this strip are of the form $C^1_{i_1}, C^2_{i_2}, \dots, C^k_{i_k}$ and in particular, the pair $A_i$ is mapped to the pair $C^i_{j}$ for some $j$.
	Let $\phi\colon [k] \to \mathbb{N}$ be the mapping such that $A_i$ is mapped precisely to $C^i_{\phi(i)}$ for every $i \in [k]$.
	
	The same argument can be applied to the pairs  $B_1, \dots, B_k$ and we define a mapping $\phi'\colon [k] \to \mathbb{N}$ such that $B_i$ is mapped precisely to $D^i_{\phi'(i)}$  for every $i \in [k].$
	
	Recall that there is a point in $\pi$ in the $(A_i, B_i)$-cell for every $i \in [k]$.
	The only non-empty cells in $\tau$ between vertices of the same color in $H$ are on the diagonal and thus, we have that $\phi(i) = \phi'(i)$ for every $i \in [k]$.
	It remains to verify that $\{v^i_{\phi(i)}, v^j_{\phi(j)}\} \in E_H$ whenever $\{i,j\}$ is an edge in $G$.
	For every $\{i,j\} \in E_G$, there is a point in the $(A_i, B_j)$-cell in $\pi$ that must be mapped to a point in the $(C^i_{\phi(i)},D^j_{\phi(j)})$-cell in~$\tau$.
	Therefore, the $(C^i_{\phi(i)},D^j_{\phi(j)})$-cell is non-empty and we have $\{v^i_{\phi(i)}, v^j_{\phi(j)}\} \in E_H$.
\end{proof}

To demonstrate the usefulness of \cref{thm:appm-hard}, let us use it to derive the hardness of pattern counting in an elementary way.
The original proof by Berendsohn, Kozma and Marx~\cite{BerendsohnKM21} uses a similar idea but it requires a more complicated application of the inclusion-exclusion principle to reduce from \textsc{Partitioned PPM}.

\begin{theorem}[{\cite[Theorem 4]{BerendsohnKM21}}]
	\label{thm:sppm-hard}
There is no algorithm that counts the number of copies of a given $k$-pattern in time $f(k)\cdot n^{o(k/ \log k)}$ for any function~$f$, unless ETH fails.
\end{theorem}
\begin{proof}[Alternative proof of \cref{thm:sppm-hard}]
	We show that an algorithm for exact counting in time $f(k) \cdot n^{o(k /\log k)}$ could be used to design an algorithm deciding \LPPM in time $g(k) \cdot n^{o(k /\log k)}$.
	We conclude that an algorithm with such running time would refute ETH via \cref{thm:appm-hard}.
	
	Suppose $(\pi, \tau)$ is an instance of \LPPM{}.
	Let $\#\pi(\tau)$ denote the number of $\pi$-copies in~$\tau$ and
	let $\tau'$ be the permutation obtained from~$\tau$ by deleting its leftmost point.
	Every copy of~$\pi$ is either left-aligned or not and therefore, the number of left-aligned $\pi$-copies is exactly $\#\pi(\tau) - \#\pi(\tau')$.
	It follows that we can compute the number of left-aligned copies by invoking the algorithm for exact counting twice and in particular, we can decide whether there is a left-aligned $\pi$-copy in $\tau$ or not.
	Therefore, we obtain an algorithm solving \LPPM{} in time $2 \cdot f(k) \cdot n^{o(k/\log k)}$.
\end{proof}

\section{Inapproximability of counting}
\label{sec:inapprox}

We first show that \LPPM can be reduced to instances of pattern counting where the text either contains very many or very few copies of the pattern.
Formally, we consider for any fixed $\eps$ such that $0 < \eps < 1/2$ the \GapSPPM problem.
The input to \GapSPPM consists of a $k$-pattern~$\pi$ and a text $\tau$ of length~$n$, and the task is to distinguish whether the number of $\pi$-copies in $\tau$ is at least $n^{(1-\eps)k}$ or at most $n^{\eps k}$.

\begin{theorem}
\label{thm:gap-hard}
For arbitrary $0 < \eps < 1/2$, \GapSPPM is \Wone-hard and cannot be solved in time $f(k)\cdot n^{o(k/ \log k)}$ for any function~$f$, unless ETH fails.
\end{theorem}
\begin{proof}
For a fixed $\eps$, we describe a parameterized reduction from \LPPM.
Let $(\pi, \tau)$ be an instance of \LPPM and set $\alpha = \lceil 2/\eps \rceil$.

\begin{figure}
\begin{subfigure}[c]{.45\textwidth}
\begin{subfigure}[T]{\textwidth}
\centering
\includegraphics[scale=0.9]{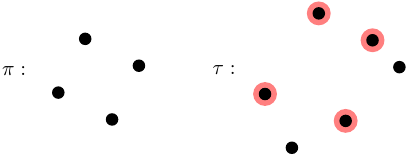}
\vspace{0.05in}
\caption{Instance $(\pi, \tau)$ of \LPPM with highlighted left-aligned $\pi$-copy in $\tau$.}
\label{fig:gap-hard-instance}
\end{subfigure}
\begin{subfigure}[T]{\textwidth}
\centering
\vspace{0.2in}
\includegraphics[scale=0.9]{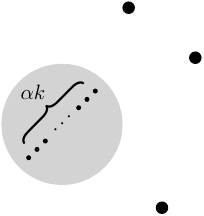}
\vspace{0.05in}
\caption{The produced pattern $\pi'$.}
\label{fig:gap-hard-pattern}
\end{subfigure}
\end{subfigure}
\hfill
\begin{subfigure}[c]{.49\textwidth}
	\centering
	\includegraphics{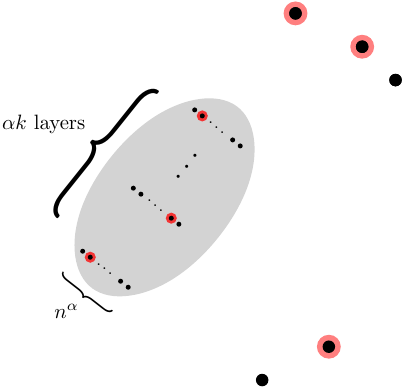}
	\vspace{0.05in}
	\caption{The produced text $\tau'$. One possible $\pi'$-copy corresponding to the left-aligned $\pi$-copy in~$\tau$ is highlighted.}
\label{fig:gap-hard-text}
\end{subfigure}
\hfill
	\caption{Illustration of the reduction in \cref{thm:gap-hard}. The initial blocks in both permutations $\pi'$ and $\tau'$ of the produced instance are highlighted.}
	\label{fig:gap-hard}
\end{figure}

First, if $n < ((\alpha + 1) \cdot k)^{2\alpha/\eps}$, we simply count the number of left-aligned $\pi$-copies exactly and return a trivial yes- or no-instance of \GapSPPM.
Observe that this takes only $f(k)$ time for some computable function $f$, even when we enumerate over all $k$-tuples in~$\tau$.

Otherwise, we take $\pi'$ to be the permutation obtained from~$\pi$ by inflating its leftmost element with an increasing sequence of length~$\alpha k$, and we let $\tau'$ be the permutation obtained from~$\tau$ by inflating its leftmost element with layered permutation consisting of $\alpha k$ layers, each of length~$n^\alpha$.
We refer to the parts obtained by inflating the leftmost elements in both $\pi'$ and $\tau'$ as their \emph{initial blocks}.
See Figure~\ref{fig:gap-hard}.

We denote by $k'$ and $n'$ the length of $\pi'$ and $\tau'$ respectively.
Observe that $k'$ is equal to $\alpha k + (k-1) \le (\alpha + 1) \cdot k$, i.e., linear in the length of $\pi'$, and $n' = n-1 + \alpha k \cdot n^\alpha$, i.e., polynomial in the size of the input instance.
We will later need the following lower and upper bounds on $n'$
\begin{equation}\label{ineq:k-and-n}
 n^{\alpha} \le n' \le (\alpha + 1) \cdot k \cdot n^{\alpha} \le n^{\eps / (2 \alpha)} \cdot n^{\alpha}
\end{equation}
where the last inequality holds since we guaranteed that $n \ge ((\alpha + 1) \cdot k)^{2\alpha/\eps}$.

\subparagraph{$(\pi, \tau)$ is a yes-instance.} First, assume that $(\pi, \tau)$ is a yes-instance of \LPPM, i.e., there is a left-aligned $\pi$-copy in $\tau$.
In this case, we obtain many corresponding $\pi'$-copies in $\tau'$ since we can choose for each of the first $\alpha k$ elements one out of $n^\alpha$ options (see Figure~\ref{fig:gap-hard-text}).
In total, there are at least $(n^\alpha)^{\alpha k} = n^{\alpha^2 k}$ such embeddings.
We now lower-bound the number of $\pi'$-copies by expressing it in terms of $n'$ and $k'$.
The number of $\pi'$-copies in $\tau'$ is at least
\begin{multline}\label{eq:gap-yes}
n^{\alpha^2 k} \ge \left(\frac{n'}{n^{\eps / (2 \alpha)}}\right)^{\alpha k} \ge
n^{-\frac{\eps}{2} \cdot k} \cdot \left(n'\right)^{\frac{\alpha}{\alpha+1}k'} >
(n')^{-\frac{\eps}{2} \cdot k'} \cdot \left(n'\right)^{\left(1-\frac{1}{\alpha+1}\right)\cdot k'} \\\ge
(n')^{-\frac{\eps}{2}\cdot k'} \cdot \left(n'\right)^{\left(1-\frac{\eps}{2}\right)\cdot k'}
= \left(n'\right)^{\left(1-\eps\right)\cdot k'}
\end{multline}
where the first inequality follows from the upper bound on $n'$ in~\eqref{ineq:k-and-n}, the second inequality is obtained by using the bound $k \ge  \frac{k'}{\alpha + 1}$ in the exponent, the third follows since $n^k < (n')^{k'}$,  and the penultimate inequality holds since $\alpha + 1 \ge \left\lceil \frac{2}{\eps} \right\rceil + 1 \ge \frac{2}{\eps}$.

\subparagraph{$(\pi, \tau)$ is a no-instance.}Conversely, assume that $(\pi, \tau)$ is a no-instance of \LPPM.
This does not necessarily imply that $\tau'$ completely avoids $\pi'$.
However, we claim that no $\pi'$-copy can contain a point from the initial block of~$\tau'$.
Intuitively, using any element of the initial block would force the existence of a left-aligned $\pi$-copy in $\tau$, contradicting the assumption.

First, let us show that every $\pi'$-copy in $\tau'$ contains at most $\alpha k$ elements from the initial block of~$\tau'$.
For contradiction, assume that there is an embedding of $\pi'$ in~$\tau'$ that maps at least $\alpha k + 1$ leftmost elements of $\pi'$ to the initial block of $\tau'$.
The leftmost $\alpha k$ elements of $\pi'$ form an increasing sequence and thus, they must be mapped to pairwise different layers in the initial block of~$\tau'$.
In particular, the element $\pi'_{\alpha k}$ is mapped to the last layer and, as a consequence, the same must be true for $\pi'_{\alpha k +1}$.
It follows that $\pi'_{\alpha k + 1}$ is mapped to a point to the right and below the image of $\pi'_{\alpha k}$ but above the image of $\pi'_{\alpha k - 1}$.
This is not possible since both $\pi'_{\alpha k - 1}$ and $\pi'_{\alpha k}$ belong to the inflation of the leftmost point in~$\pi$ while $\pi'_{\alpha k +1}$ does not.
It is important here that our choice of $\alpha$ implies $\alpha k \ge 2$.

Now let us assume that there is a $\pi'$-copy in $\tau'$ that contains at least one element of the initial block.
We already know that this copy uses at most $\alpha k$ elements from the initial block of~$\tau'$ and therefore, we would obtain a left-aligned $\pi$-copy in $\tau$ by deflating the initial blocks in both $\pi'$ and $\tau'$ back to single elements.
Therefore, every $\pi'$-copy in~$\tau'$ completely avoids the initial block of~$\tau'$ and the total number of copies is at most
\begin{equation}\label{eq:gap-no}
\binom{n-1}{k'} \le n^{k'} \le (n')^{\frac{k'}{\alpha}} < (n')^{ \eps \cdot k'}
\end{equation}
where the first inequality is trivial, the second inequality follows from \eqref{ineq:k-and-n} and the final inequality holds since $\frac{1}{\alpha} \le \frac{\eps}{2} < \eps$.

\subparagraph{}
We have shown that if there is a left-aligned $\pi$-copy in $\tau$, there are at least $(n')^{(1-\eps)\cdot k'}$ $\pi'$-copies in $\tau'$ by~\eqref{eq:gap-yes} and otherwise, there are at most $(n')^{ \eps \cdot k'}$ such copies by~\eqref{eq:gap-no}.
Therefore, an algorithm deciding \GapSPPM in time $f(k) \cdot n^{o(k / \log k)}$ can be used to decide \LPPM in time $f(k') \cdot (n')^{o(k' / \log k')}$ which would refute ETH through \cref{thm:appm-hard} since $k' \in \fO(k)$ and $n' \le (\alpha + 1) \cdot k \cdot n^{\fO(1)}$.
\end{proof}

Finally, we can prove \cref{thm:inapprox} by showing that any $n^{(1/2-\eps)\cdot k}$-approximation algorithm can be used to decide $(n^{\eps' k}, n^{(1-\eps') \cdot k})$\textsc{-GAP \#PPM} for some positive $\eps'$.

\mainresult*
\begin{proof}
We assume that such an approximation algorithm exists for some $\eps > 0$ and we show that \GapSPPMh can be decided by its single invocation.
On a given instance~$(\pi, \tau)$, we simply run the approximation algorithm and answer positively if and only if it reports that the number of $\pi$-copies in $\tau$ is greater than $n^{\frac{1}{2}k}$.

If $(\pi, \tau)$ is a yes-instance of \GapSPPMh, there are at least $n^{(1 - \eps/2) \cdot k}$ copies of~$\pi$ in~$\tau$.
Thus, the number of $\pi$-copies reported by the approximation algorithm is at least
\[\frac{n^{\left(1 - \frac{\eps}{2}\right) \cdot k}}{n^{\left(\frac12-\eps\right) \cdot k}} = n^{\left(\frac12 + \frac{\eps}{2}\right) \cdot k} > n^{\frac12 k}. \]

Otherwise, $(\pi, \tau)$ is a no-instance with at most $n^{(\eps/2) \cdot k}$ copies of $\pi$ in $\tau$ and the output of the algorithm is at most
\[n^{\left(\frac{\eps}{2}\right) \cdot k} \cdot n^{\left(\frac12-\eps\right)\cdot k} = n^{\left(\frac12 - \frac{\eps}{2}\right) \cdot k} < n^{\frac12 k}. \]

It follows that the assumed approximation algorithm can be used to correctly decide \GapSPPMh in the exact same runtime.
Therefore, an $n^{(1/2-\eps)\cdot k}$-approximation algorithm running in $f(k) \cdot n^{o(k/\log k)}$ time would refute ETH through \cref{thm:gap-hard}.
\end{proof}

Finally, we remark that analogous arguments rule out the existence of $f(k)\cdot n^{o(k/ \log k)}$-time approximation algorithms with one-sided multiplicative error of $n^{(1-\eps) k}$ for every $\eps > 0$.

\bibliography{main}
\end{document}